\documentclass{article}
\usepackage{amssymb}
\usepackage{graphicx}
\usepackage{amsmath}
\usepackage{hyperref}
\usepackage{bm}
\usepackage{MnSymbol}
\usepackage[toc,page]{appendix}
\newtheorem{theorem}{Theorem}

\newtheorem{definition}[theorem]{Definition}

\newtheorem{proposition}[theorem]{Proposition}

\newenvironment{proof}[1][Proof]{\textbf{#1.} }{\ \rule{0.5em}{0.5em}}

\DeclareMathAlphabet{\mathpzc}{OT1}{pzc}{m}{it}

\newcommand {\id}{\mathrm{id}}

\renewcommand {\L}{{\cal L}}
\newcommand {\M}{{\cal M}}

\newcommand {\Q}{{\cal Q}}

\newcommand {\G}{\mathcal{G}}

\renewcommand {\O}{\mathcal{O}}

\newcommand{\susy}{\mathfrak{susy}}

\newcommand{\Tor}{\mbox{{\rm Tor}}}

\newcommand{\Sym}{\mathrm{Sym}}

\newcommand{\GL}{\mathrm{GL}}
\newcommand{\Spin}{\mathrm{Spin}}
\newcommand{\SO}{\mathrm{SO}}
\newcommand{\g}{\mathfrak{g}}

\newcommand{\dbar}{\bar \partial}
\newcommand{\E}{\mathcal{E}}
\newcommand{\so}{\mathfrak{so}}

\newcommand{\ity}{_{\infty}}
\newcommand{\Ad}{\mathrm{Ad}}
\newcommand{\OGr}{\mathrm{OGr}}
\newcommand{\wOGr}{\widetilde{\mathrm{OGr}}}

\newcommand{\p}{\mathfrak{p}}

\newcommand{\gl}{\mathfrak{gl}}
\renewcommand{\sl}{\mathfrak{sl}}

\renewcommand{\P}{\mathbf{P}}
\newcommand{\SL}{\mathrm{SL}}

\newcommand{\U}{\mathrm{U}}
\newcommand{\ddef}{\mathrm{def}}
\newcommand{\s}{\mathfrak{s}}
\renewcommand{\t}{\mathfrak{t}}
\newcommand{\sd}{\partial}

\renewcommand{\H}{\mathrm{H}}
\newcommand{\vv}{V}

\newcommand{\C}{\mathbb{C}}
\newcommand{\R}{\mathbb{R}}
\newcommand{\cP}{\mathcal{P}}
\newcommand{\cR}{\mathcal{R}}
\newcommand{\red}{\mathrm{red}}
\newcommand{\rA}{\mathrm{A}}
\newcommand{\rB}{\mathrm{B}}
\newcommand{\Grv}{\Theta}
\newcommand{\fu}{\mathfrak{u}}

\begin{document}
\title{The odd twistor transform in eleven-dimensional supergravity} 
\author{M.V. Movshev\\Stony Brook University\\Stony Brook, NY, 11794-3651,\\ USA \\ \texttt{mmovshev@math.sunysb.edu}} 

\maketitle

\begin{abstract}
We define a twistor-like transform of the  equations of  eleven-dimensional supergravity.  More precisely these equations are encoded by the CR-structure on the twistor space $\cP^{2\times15+11|8\times2+16}$.  In addition equations of the linearized eleven-dimensional supergravity adapted to the 3-form potential  can be transformed  into the tangential Cauchy-Riemann equation $\dbar \omega=0$ on  $\cP$. 
\end{abstract}


\tableofcontents
\section{Introduction}
The classical superspace formulation (\cite{BrinkHowe}, \cite{CremmerFerrara}) makes the supersymmetries manifest, with  a drawback that the fields it encodes  are constrained to satisfy supergravity equations. The proposal  of Cederwall \cite{Cederwall} is supposed to rectify that. Cederwall's approach still depends on the choice of a  background solution of supergravity equations, but the fields in his construction are unconstrained. In  the flat background, the fields   are elements of 
\begin{equation}\label{E:gr}
\Grv=\Grv_{\mathbb{R}^{11|32}}=A\otimes \Lambda[\theta^{1}, \dots,\theta^{32}]\otimes C^{\infty}(\mathbb{R}^{11}).
\end{equation}

The tensor factor $A$ is the  commutative algebra 
\begin{equation}\label{E:Adef}
\mathbb{R}[\lambda^{1},\dots,\lambda^{32}]/(v^i)
\end{equation}
with 
\begin{equation}\label{E:pure}
v^i=\Gamma^{i}_{\alpha\beta}\lambda^{\alpha}\lambda^{\beta}, i=1,\dots, 11, \alpha,\beta=1,\dots,32
\end{equation}
The algebra is graded by the degree in $\lambda^{\alpha}$.
We use  summation convention over repeated indices:
lower-case Greek letters run between $1$ and $32$, lower-case Roman letters have a range between $1$ and $11$, capital  Greek letters  run between $1$ and $8$.
In addition to  variables $\lambda^{\alpha}\in A$ and $\theta^{\alpha}\in  \Lambda[\theta^{1}, \dots,\theta^{32}]$ it is convenient to fix   and coordinates $x^1,\dots,x^{11}$ on $\R^{11}$. 
 The variables $\lambda^{\alpha},\theta^{\alpha}$ transform in a  spinor representation $\s^{\mathbb{R}}_{10,1}$   of the Lorentz group $\Spin(10,1,\mathbb{R})$. The coordinates  $x^i$ on $V^{10,1\ \R}$ transform as    vectors under  $\SO(10,1,\mathbb{R})$. The reader might wish to consult   \cite{Deligne}  for a mathematical introduction to spinors and $\Gamma$-matrices in the Lorentz signature.

The operator 
\begin{equation}\label{E:defd}
D=\lambda^{\alpha}\eta_{\alpha}
\end{equation}
where
\begin{equation}\label{E:vect}
\eta_{\alpha}=\frac{\sd}{\sd \theta^{\alpha}}-\Gamma^i_{\alpha\beta}\theta^{\beta}\frac{\sd}{\sd x^i}
\end{equation}
is a differential in $\Grv$.
According to \cite{CGNN},\cite{CNT},\cite{MemBerkovits} the $D$-cohomology of $\Grv$ coincide with the space of solutions of the linearized equations of eleven-dimensional supergravity in the flat background. 

In order to get a better grasp of the constructions from \cite{Cederwall},  it is desirable to identify  $\Grv$ with a construction  already  known in homological algebra, e.g. complexes of DeRham, Dolbeault, tangential CR complex and such. 

In this paper we construct a manifold $\cP$ which we call the odd twistor transform of  the D=11 supergravity equations (SUGRA). 
This is a super CR-manifold (see \cite{RoslyandSchwarz} for an introduction to CR structures for physicists). We establish a quasi-isomorphism of $\Grv$ and the tangential CR complex of $\cP$.

We emphasize  that CR structures are ubiquitous  in twistor theory \cite{HoweandHartwell} and that their appearance in our work is not surprising.  What is  unexpected is that our twistor transform encodes solutions of SUGRA rather than equations of  conformal supergravity or equations of  anti self-duality, as the conventional  (ambi-)twistor constructions do (cf. \cite{WardandWells}). 

 The odd twistor transform is a modification of the superspace gravity of Brink and Howe \cite{BrinkHowe}.  In this approach a solution of SUGRA on a Lorentz oriented spin-manifold $M^{11}$ is  encoded by a super-extension $\M=\M^{11|32}$ and a rank $(0|32)$ subbundle 
\begin{equation}\label{E:defdist}
F\subset T(\M)
\end{equation} of the tangent bundle. The manifold $\M=\Pi S$ is the total space of the spinor bundle over $M$ with the parity of fibers reversed. 
 If the vector fields $\xi_{1},\xi_{2}$ are  in $F$,  the commutator $[\xi_{1},\xi_{2}]$ might   not be. 
 The Frobenius tensor, or the torsion $T(\xi_{1},\xi_{2})$, is the normal component of $[\xi_{1},\xi_{2}]$. 
$T$ is a map of  vector bundles
\begin{equation}\label{E:torsion}
T:\Lambda^2F\rightarrow N=T(\M)/F
\end{equation}
Bear in mind that since $F$ is odd, $T$ is symmetric.
By the results of \cite{BrinkHowe}, the equations of SUGRA can be  written succinctly as
\begin{equation}\label{E:SG}
T_{\alpha\beta}^i=\Gamma_{\alpha\beta}^i
\end{equation}
where $T_{\alpha\beta}^i$ is the matrix of $T$ in suitable local bases. 
For example, in  flat space-time $\mathbb{R}^{11|32}$ the subbundle $F$ is spanned by vector fields (\ref{E:vect}).
The  odd subbundle $F$ has symmetric inner product $C(\cdot,\cdot)$, which has a skew-symmetric matrix $C_{\alpha\beta}$. The normal bundle $N$ carries a Lorentz metric $g_{ij}$.
 Identity $\det(\Gamma_{\alpha\beta}^iy_i)=q(y)^{16}$ where 
\begin{equation}\label{E:metric}
q(y)=g^{ij}y_{i}y_j
\end{equation}
 defines the  conformal class  $g_{ij}$ (see  \cite{BrinkHowe}  for an explanation of 
 how to fix  $g_{ij}$ in its conformal class and  how to define the  odd symmetric pairing $C_{\alpha\beta}$ on $F$). \footnote{In this note we shall not make a distinction between an orthonormal basis and  a weak orthonormal basis in $N$, that is a collection of  sections $f_i\in N$ such that $g(f_i,f_j)=\pm c \delta_{ij},c>0$. The last notion makes sense when the metric  is defined only up to a scaling factor.}  In the flat case $C_{\alpha\beta}$ comes from the symplectic $\Spin(10,1,\mathbb{R})$-invariant form on $\s^{\mathbb{R}}_{10,1}$.  To define  $\Gamma^{\alpha,\beta, i_1,\dots,i_k}$, we use $g_{ij}$, $C_{\alpha\beta}$ and  (\ref{E:SG}).

The  {\it odd twistor transform}
$\cP=\cP_{\M,F}$ of the SUGRA datum $(\M,F)$ is a  relative Isotropic Grassmannian $\OGr(2,11)_{\M}$. Locally as a real manifold it is a product 
\begin{equation}\label{E:coset}
\R^{11|32}\times \OGr(2,11)=\R^{11|32}\times \SO(11,\R)/\U(2)\times \SO(7,\R).
\end{equation} The group $\SO(n,\R)$ is a compact form of $\SO(n)\overset{\ddef}{=}\SO(n,\C)$.    A point in  $\OGr(2,11)$ is represented by a light-like (or isotropic) two-plane $W$ in the complexified Minkowski space. More formally this can be said as follows. The complexification $V^{11}$ of $\vv^{10,1\ \R}$ is equipped with the complexified inner product $(\cdot,\cdot)$. As an  algebraic variety $ \OGr(2,11)$ is isomorphic to \[\{W\subset \vv^{11}|\dim W=2, (\cdot,\cdot)|_{W}=0 \}.\] To get an equivalent description of  $\OGr(2,11)$, we can  consider  isotropic two-planes in  the complexification of an  Euclidean eleven-dimensional space. This is  unnatural from the standpoint of the gravity theory, but explains why  $\OGr(2,11)$ is a coset space (\ref{E:coset}) and makes the topological structure of $\OGr(2,11)$ more apparent. The superspace $\OGr(2,11)_{\M}$ can be embedded in the projective bundle $\P(\Lambda^2N^{\mathbb{C}})$ using the Pl\"ucker embedding. $\OGr(2,11)_{\M}$   is defined fiber-wise by  equations (\ref{E:Pluckercoord}), which are  written in a local $g_{ij}$-orthonormal basis of $N$. 
The space $\cP=\OGr(2,11)_{\M}$ has  real dimension $(2\times15+11|2\times 8+16)$. 

The manifold $\cP$ has a CR-structure (Definition \ref{D:CR}) defined by means of the complex subbundle of the complexified tangent bundle $H^{1,0}\subset T^{\mathbb{C}}({\cP})$. We begin the explanation of its construction with a remark that the fibers of the projection \begin{equation}\label{E:projectionTwistor}
p:{\cP}\rightarrow \M
\end{equation}
are complex manifolds, which are isomorphic to $\OGr(2,11)$. The space $H^{1,0}_x,x\in {\cP}$, is characterized by the  condition that  the kernel of the  differential  $Dp$ 
\begin{equation}\label{E:twmap}
H^{1,0}_x\overset{Dp}\rightarrow  T^{\mathbb{C}}_{z},\quad z=p(x)
\end{equation}
is the complex tangent space $T_{x}(\OGr(2,11))=T^{vert}_{x} $ to the fiber $p^{-1}(z)$ at $x$. 
The image  $(Dp)H^{1,0}_x$  is spanned by the complex vector fields
\begin{equation}\label{E:distgengen}
\xi_{\beta}=\bar{a}^{ij}\Gamma_{\beta ij}^{\alpha }\eta_{\alpha}.
\end{equation}
The variables $a^{ij}$ are the Pl\"ucker  coordinates (\ref{E:Plucker}) of the isotropic two-plane $W$ corresponding to  $x\in p^{-1}(z)\cong \OGr(2,11)$, and  $\{\eta_{\alpha}\}$ is a basis in $F_{z}$. 
Complex conjugation on $\Lambda^2N^{\mathbb{C}}$ defines an involution $\rho$ on  $\cP$. 

Here is  our first result about $\cP$.
\begin{proposition}
Let $(\M,F)$ be a real $(11|32)$-dimensional supermanifold, such that the Frobenius tensor of rank $(0|32)$distribution $F$
 satisfies (\ref{E:SG}).
 Then the CR structure $H^{1,0}$ given by (\ref{E:twmap}) on the relative Isotropic Grassmannian $\cP_{\M,F}=\OGr(2,11)_{\M}$ is integrable. The complex involution $\rho$ on $\cP_{\M,F}$ maps $H^{1,0}$ to $H^{0,1}$ and  $H^{1,0}\cap H^{0,1}=\{0\}$.
 \end{proposition}
 See  Section \ref{E:P} for the  proof and discussion.
 The inverse transform $\cP\Rightarrow (\M,F)$ is  defined if $\cP$ satisfies conditions of the following theorem (see Section \ref{S:inverse} for details).
\begin{proposition}
Let $\cP$ be a globally embeddable  (see Definition \ref{D:embeddable}) $(2\times15+11|2\times8+16)$-dimensional super CR manifold. Suppose that $\cP$ satisfies conditions (\ref{E:data1},\ref{E:data2},\ref{E:data2.5},\ref{E:data3}) in Section \ref{S:inverse}. Then $\cP$ is isomorphic to the odd twistor transform of some $(\M,F)$.
\end{proposition}
In this Proposition conditions (\ref{E:data1},\ref{E:data2},\ref{E:data2.5},\ref{E:data3})  seem to be essential conditions. It is desirable to get rid of the global embeddability because it is not intrinsic to the CR nature of the problem. 

 {\it The tangential Cauchy-Riemann  complex} (cf. \cite{Boggess} and (\ref{E:CRdef})) $ \Omega_{H^{0,1}}= \bigoplus_{p\geq 0} \Omega^{0,p}_{H^{0,1}}$ is an  analogue of the Dolbeault complex for  CR (super)manifolds.
A generalization $\Grv_{\M,F}$ of the complex (\ref{E:gr}) can be defined for a non-flat space-time $\M$ and a distribution $F$ (see Section \ref{S:structure} and \cite{BonoraandPastiandTonin}, \cite{BerkovitsandHoweint} for  details and further development).

The map  (\ref{E:projectionTwistor}) induces a homomorphism of differential graded algebras 
\begin{equation}\label{E:ldefm}
p_{H_{0,1}}^*:\Grv_{\M,F}\rightarrow \Omega_{H^{0,1}}
\end{equation} 
(see Section \ref{S:structure} for details). 
Note that  $\cP_{\M,F}$ has smaller dimension  than the space underlying $\Theta_{\M,F}$. In this sense $\cP_{\M,F}$ gives a more economical description of SUGRA.

Our main result is the comparison of the cohomologies of $\Grv_{\M,F}$ and $\Omega_{H^{0,1}}(\cP)$ (see the end of Section \ref{invariantCRstr} and Section \ref{S:reality}):
\begin{proposition}\label{P:quasiisomorphism}The map $p_{H_{0,1}}^*$ defines an isomorphism between the $D$-cohomology of $\Grv_{\M,F}$ and the $\rho^*$-real $\dbar_{H^{0,1}}$-cohomology  of  $\Omega_{H^{0,1}}(\cP)$. \end{proposition}

Recall that $D$-cohomology of $\Grv_{\M,F}$ have an interpretation of solutions of linearized equations of SUGRA.
It would be interesting in light of this result to explore the possibility of a formulation of the full nonlinear equations   on $\cP$. 

We conclude the introduction with a list of related problems.
\begin{enumerate}
\item The action  of  linearized gravity theory in the pure spinor approach \cite{MemBerkovits} has the form $S=\int d^{11}x\langle\Psi Q\Psi\rangle$,  where the norm $\langle \rangle$ is  such that $\langle\lambda^7\theta^9 \rangle=1 $. Proposition \ref{P:quasiisomorphism} can be interpreted as a statement about an isomorphism of the space of solutions of the equations of  linearized supergravity $Q\Psi=0$   and space of solutions of $\dbar_{H^{0,1}}f=0, f\in \Omega_{H^{0,1}}(\cP)$.  It is plausible  that  $p_{H_{0,1}}^*$ defines an equivalence of the actions $S$ and $\int_{\cP} d\mu f\dbar_{H^{0,1}} f $, where $d\mu$ is some integral volume form on $\cP$ derived from the norm $\langle \rangle$.    It would be interesting to find $d\mu$ using ideas of \cite{MovshevQ}, \cite{MasonandSkinner}.
The next problem is closely related.

\item The work \cite{Cederwall} gives a description of the supergravity Lagrangian $\L_{\mathrm{SUGRA}}$ in a superspace formulation with  auxiliary pure spinor fields. Some of the terms of $\L_{\mathrm{SUGRA}}$ (such as $\Gamma_{\alpha\beta }^{ij}\lambda^{\alpha}\lambda^{\beta}$) are sections of a bundle  
on $\OGr(2,11)$. It is tempting to speculate that the Lagrangian can be defined   on  $\cP$.   The idea is to interpret the bracket  $\{,\}$ defined by the formula  \[\Psi\{\Psi,\Psi\}  \overset{\ddef}{=}  \lambda\Gamma_{ab}\lambda\Psi R^{a}\Psi R^{b}\Psi, \] taken from  the full supergravity Lagrangian \cite{Cederwall} as a weak Poisson structure (a G$\ity$-structure with a trace in the mathematical slang). If this guess is correct,  the technique of \cite{MArkl} can be used to transfer the G$\ity$-structure to $\Omega_{H^{0,1}}(\cP)$.
\item  SUGRA is a low energy limit of M-theory. It is believed that M-theory properties are related to the supermembrane  \cite{HughesandLiuandPolchinski} \cite{BergshoeandSezginandTownsend}\cite{DuffandoweandInamiandStelle}. Pure spinors play a fundamental role in the covariant formulation of the supermembrane \cite{MemBerkovits}. 
It is interesting to translate supermembrane from the superspace to the twistor space. One of the attractive feature of twistors is that 
the polynomials $a^{ij}=\Gamma_{\alpha\beta}^{ ij}\lambda^{\alpha}\lambda^{\beta}$ after the blowup become  basic generators \cite{MovGr}.  The nonlinear constraint  $\lambda\Gamma_{ ij}\lambda \Pi_{J}^j=0$ \cite{MemBerkovits}, where $\Pi_{J}^j$  is the canonical momentum,  simplifies to
\begin{equation}\label{quastion3}
a^{ij}\Pi_{J}^j=0.
\end{equation} 
It would be interesting to systematically apply the odd twistor transform  to the supermembrane and its double reduction - strings.
\end{enumerate}
We plan to address these questions in future publications.

Here is an outline of the structure of the paper.  In Section \ref{E:P} we establish integrability  of the CR structure of $\cP$.  In Section \ref{S:structure} we define the tangential CR complex $\Omega_{H^{0,1}}(\cP)$ and  the non-flat generalization  $\Grv_{\M,F}$  of  the complex (\ref{E:gr}). In the same section we also define the map $p^*_{H^{0,1}}$ between these complexes. Reality conditions, which are used later to characterize physical fields,  are formulated in Section \ref{S:reality}. It is known that not every  holomorphic supermanifold admits a projection onto the underlying manifold. Supermanifolds having this property are called split. 
In Section \ref{invariantCRstr} we define  obstruction of being split in the context of CR manifolds that are odd twistor transforms. We also establish that the map $p^*_{H^{0,1}}$ defines an isomorphism on cohomology. In Section \ref{S:inverse} we invert the odd twistor transform under certain assumptions of analyticity. 
Section \ref{S:modCR} briefly describes an interesting even modification of the CR structure on $\cP$.  The appendices contain discussion of some  technical points. 
In particular, in Appendix \ref{S:adjoint} we justify the local description of the map $p^*_{H^{0,1}}$. 
The Pl\"ucker embedding of $\OGr(2,11)$ is characterized by equations in Appendix \ref{S:Plucker}. 
Orbits of $\SO(10,1,\R)$ in $\OGr(2,11)$ are listed in Appendix \ref{S:orbits}. The super-Poincar\'e group acts on the odd twistor transform $\cP$ of the flat solution of SUGRA.  The group preserves the CR structure and  has a dense orbit in $\cP$.  In Appendix \ref{S:flat} we give a  Lie algebraic description of the CR structure on this orbit, considered as a homogenous space.

The author would like to thank P. Howe, D. Hill, C. LeBrun, A. Schwarz and W. Siegel  for  useful comments.


\section{Integrability of the CR structure on  $\cP_{\M,F}$}\label{E:P}
We devote this section to the proof of integrability of the CR structure  on $\cP_{\M,F}$. But first we give a formal definition of the  CR structure.
 \begin{definition}\label{D:CR}(cf. \cite{Boggess})
Let $Y$ be a $C^{\infty}$ super-manifold, equipped with a subbundle $H^{1,0}$ of the complexified tangent bundle $T^{\mathbb{C}}=T^{\mathbb{C}}(Y)$.  If $H^{1,0}\cap \overline{H^{1,0}}=0$(we shall call it a nondegeneracy condition), then $Y$ is a  Cauchy-Riemann (CR) manifold.
 If the space of sections in $H^{1,0}$ (or in $H^{0,1}=\overline{H^{1,0}}$) is closed under the commutator (we shall call it an  involutivity condition), then the CR structure  is integrable.
\end{definition}

Verification of the nondegeneracy condition is done in \cite{MovGr}.
Let us  check integrability of $H^{1,0}$. The vector fields $\xi_{\beta}$  commute with the local vertical holomorphic vector fields  in notations of  (\ref{E:distgengen}).
Locally we decompose  the tangent bundle $T(\M)$ into a direct sum $F+N$. With  a suitable choices of local bases $\{\eta_{\alpha}\}$ of $F$ and $\{\upsilon_i\}$ of $N$ the commutators $[\eta_{\alpha},\eta_{\beta}]$ decompose into $\eta_{\alpha\beta}+\Gamma_{\alpha\beta}^i\upsilon_i$, where $\eta_{\alpha\beta}$ are some sections of $F$.  The commutator of the vector fields $\xi_{\gamma}$ (\ref{E:distgengen}) is 
\begin{equation}\label{E:commutator}
[\xi_{\gamma},\xi_{\delta}]
=\bar{a}^{ij}\Gamma_{\gamma ij}^{\alpha }\left(\bar{a}^{kl}\Gamma_{\delta kl}^{\beta }\eta_{\alpha\beta}\right)+\bar{a}^{ij}\Gamma_{\gamma ij}^{\alpha }\bar{a}^{kl}\Gamma_{\delta kl}^{\beta }\Gamma_{\alpha\beta}^i\upsilon_i.
\end{equation}
The $N$-component has coefficients
\begin{equation}\label{correctnessofhom}
P^s(\bar{a})=\Gamma^{s}_{\alpha\alpha'} \bar{a}^{ij}\Gamma_{\beta ij}^{\alpha }\bar{a}^{kl}\Gamma_{\beta' kl}^{\alpha'}
\end{equation}
 These coefficients  are zero because $P^s({a})$ transforms as a 
 $\SO(11)$ vector. However 
 a vector representation is not a subrepresentation  of $\Sym^2[\Lambda^2 V^{11}]$\footnote{In this paper $\Sym^i E$  stands for the $i$-th symmetric  power of a representation or a vector bundle.} \cite{MovGr}. The remaining terms in (\ref{E:commutator}) are sections of $H^{1,0}$. This proves integrability.

 The involution $\rho$ from the introduction  leaves  equations  that define $\OGr(2,11)$ (\ref{E:Pluckercoord}) invariant. 
A point $W=\overline{W}\in \OGr(2,11)$ is a complexification of the light-like real plane. A set of such planes is empty  in Lorentz signature. We conclude that $\rho$ is  fixed point free on $\cP$. 
The involution turns $\bar{a}^{ij}$ into $a^{ij}$ in (\ref{E:distgengen}) and swaps $H^{1,0}$ with $H^{0,1}$.

\section{The complexes $\Omega_{H^{0,1}}(\cP)$ and $\Grv_{\M,F}$}\label{S:structure}
 In this section we define the complexes that appeared  in the introduction.  
  
Construction of the tangential Cauchy-Riemann complex is based on the observation  that the CR structure is integrable  if and only if  the ideal 
\begin{equation}\label{E:defideal}
I=\{\omega\in \Omega\otimes \mathbb{C}|\forall\xi_{i}\in H^{0,1}\quad  \omega(\xi_{1},\dots,\xi_{\deg \omega})=0\}=\bigoplus_{p\geq 0} I^{p}\subset \Omega
\end{equation}
in the algebra differential forms  $\Omega=\Omega(Y)$ is $d$-closed: $d(I)\subset I$(see e.g.  \cite{Boggess} for the proof of the even case).
A CR-form   $\omega \in$ 
\begin{equation}\label{E:CRdef}
\Omega^{0,p}_{H^{0,1}}\overset{\ddef}{=}\Omega^{p}/I^p
\end{equation} is
$\sum_{i_1\dots,i_p}\omega_{i_i\dots,i_p}\bar{\nu}^{i_1}\wedge \cdots \wedge \bar{\nu}^{i_p}$
where $\bar{\nu}^{i}$ are complex-linear functionals  on $H^{0,1}$. The tangential Cauchy-Riemann operator $\dbar=\dbar_{H^{0,1}}$ in $\Omega_{H^{0,1}}=\Omega_{H^{0,1}}(Y)=\bigoplus_{p\geq 0}\Omega^{0,p}_{H^{0,1}}(Y)$ is induced by the DeRham differential $d$.
The map of complexes
\begin{equation}\label{E:res}
res_{H^{0,1}}:\Omega\rightarrow \Omega_{H^{0,1}}
\end{equation}
is  a restriction     onto $H^{0,1}$.
In our applications we are mainly  interested in $\Omega_{H^{0,1}}(\cP)$.

Another  complex announced in the introduction is  $\Grv_{\M,F}$.  
 It is a generalization of $\Grv$ (\ref{E:gr}).  
In order to define $\Grv_{\M,F}$, we choose linearly independent  even forms $E^{i}$ that vanish on $F$. These forms  are a  part of the vielbein and  generate a locally  free subsheaf in $\Omega^1(\M)$ of rank $(11|0)$. Let $x^{i},\theta^{\alpha}$ be local coordinates on $\M$. Without a loss of generality  $E^i$ is equal to $dx^i-T^i_{\alpha\beta}(x,\theta)\theta^{\alpha}d\theta^{\beta}$. The forms $E^i$  characterize the distribution (\ref{E:defdist}).
Equality
\begin{equation}\label{E:deGamma}
d(E^i)=\Gamma_{\alpha\beta}^{i}d\theta^{\alpha}d\theta^{\beta}+E^kG_k^i,
\end{equation} where $G_k^i$ are  some one-forms, is equivalent to (\ref{E:torsion},\ref{E:SG}). This implies that forms 
\begin{equation}\label{E:Jgenerators}
E^i, \Gamma_{\alpha\beta}^{i}d\theta^{\alpha}d\theta^{\beta}
\end{equation} generate a differential ideal $J\subset \Omega(\M)$. We define $\Grv_{\M,F}$ to be $\Omega(\M)/J$.
Together  with $x^{i},\theta^{\alpha}$  variables $d\theta^{\alpha}=\lambda^{\alpha}$ are   local generators of $\Grv_{\M,F}$ - the deformed version of the algebra (\ref{E:gr}). The algebra $\Grv_{\M,F}$ is  graded by $\deg_{\lambda}$. When we say that $\Grv_{\M,F}$ is a deformation of $\Theta$ we mean that  locally only the differential $D=D_{\M,F}$ in $\Grv_{\M,F}$ gets deformed: 
\begin{equation}\label{E:differentialdeformed}
 Dx^i=T^i_{\alpha\beta}\theta^{\alpha}\lambda^{\beta}.
\end{equation}

There is an analogue of the map (\ref{E:res}) for $\Grv_{\M,F}$:
\[res_{F}:\Omega(\M)\rightarrow \Omega(\M)/J=\Grv_{\M,F}\]

Construction of the map $p_{H_{0,1}}^*$ (\ref{E:ldefm}) requires a clarification. In order to define $p_{H_{0,1}}^*(\omega)$, we pick   $\tilde\omega\in \Omega(\M)$ such that $ res_{F}\tilde\omega=\omega$. We define  $p_{H_{0,1}}^*(\omega)$ to be $res_{H^{0,1}}p^*\tilde\omega$. This is not the end of the story. We need to verify that $p^*J\subset I$. We chek this on the generators (\ref{E:Jgenerators}).    It follows immediately from the  definition of $H^{0,1}$ (\ref{E:twmap})  that $p^*E^i$ vanishes on $H^{0,1}$. As a result $p^*E^i\in I$. The identity  \[\xi_{\delta_1}\righthalfcup (\xi_{\delta_2}\righthalfcup p^*\Gamma_{\alpha\beta}^{i}d\theta^{\alpha}d\theta^{\beta})=0\] for the  vector fields $\xi_{\delta_i}$  (\ref{E:distgengen}) is true  because the polynomials (\ref{correctnessofhom}) are zero.   We conclude that  $p^*\Gamma_{\alpha\beta}^{i}d\theta^{\alpha}d\theta^{\beta}\in I$ and $p^*J\subset I$.
It implies that  (\ref{E:ldefm}) is well defined and $p_{H_{0,1}}^*$ is a map of complexes.

Our next goal is to write $p_{H_{0,1}}^*$ in local coordinates on $\M$ and $\cP$. 
Let $a^{ij}(W)$ be the  Pl\"{u}cker coordinates of $W\in \OGr(2,11)$ (see Appendix \ref{S:Plucker}).The family of vector spaces  
\begin{equation}\label{E:swdefinition}
\s_W=\{{a}^{ij}(W)\Gamma_{\beta ij}^{\alpha }\eta_{\alpha}| \eta_{\alpha}\in \s_{11}\}\subset \s_{11}\overset{\ddef}{=}\s_{10,1\ \R}\otimes \C
\end{equation} 
defines a complex vector bundle
\begin{equation}\label{E:bundledef}
\s_{\OGr(2,11)}=\{(W,\xi)|W\in \OGr(2,11), \xi \in {\s_W}\}
\end{equation}

We are going to define coordinates on the total space of $\s_{\OGr(2,11)}$ that  will be  used in  the local description of $p_{H_{0,1}}^*$. For this purpose, we need a  basis in the space of local sections of $\s_{\OGr(2,11)}$.
Such a basis 
can be seen rather explicitly. We fix a basis $\{\eta_{\alpha}\}$ in $\s_{11}$ that is compatible with the decomposition (\ref{E:grading}), such that $\eta_{1},\dots,\eta_{8}\in s^1 $, $\eta_{9},\dots,\eta_{24}\in s^0 $, and  $\eta_{25},\dots,\eta_{32}\in s^{-1}$. We pick a plane  $W\in \OGr(2,11)$ the same as  in the proof of the isomorphism \ref{E:imageW}
and  choose it to be close to $U$ in   
  (\ref{E:decomposition1}). 
  
   We pick $kl$ such that  $a^{kl}(U)\neq 0$. We set  \[\mu_{\beta}=\frac{{a}^{ij}(W)}{a^{kl}(W)}\Gamma_{\beta ij}^{\alpha }\eta_{\alpha}, \beta=25,\dots,32\]
\begin{equation}\label{E:Unondegeneracy}
\mbox{Note that when $W=U$ then $\mu_{24+\alpha}$ is proportional  to $\eta_{24+\alpha}$. }
\end{equation}
 This means that $\{\mu_{\beta}\}$ are linearly independent sections of $\s_{\OGr(2,11)}$ in a Zariski neighborhood of $U$. Let $\mu^{\rA},\rA=1,\dots,8$ be sections of the dual bundle such that 
 \begin{equation}\label{E:coord}
 \mu^{\rA}(\mu_{24+\rB})=\delta^{\rA}_{\rB}.
 \end{equation}
 A variable $\lambda^{\alpha}$ defines a linear function on fibers of $\s_{\OGr(2,11)}$ because the fibers are subspaces in $\s_{11}$. 
 It follows immediately that 
 \begin{equation}\label{E:lambdaformula}
 \lambda^{\alpha}=\sum_{\rA=1}^8\frac{{a}^{ij}(W)}{a^{kl}(W)}\Gamma_{24+\rA ij}^{\alpha }\mu^{\rA}
 \end{equation}

The locally defined CR-forms on $\cP_{\M}$ are  functions in 
\begin{equation}\label{E:variables}
x^{i},\theta^{\alpha}, a^{ij},\bar{a}^{ij}, \mu^{\rA},d\bar{a}^{ij}
\end{equation} 
that have the total $\GL(2,\C)$-scaling degree zero(\ref{E:scalingdegree}) in  $a^{ij}$, $\bar{a}^{ij}$ and $\bar{a}^{ij}$. 
The map  (\ref{E:ldefm}) keeps  $x^i,\theta^{\alpha}$ unchanged and  replaces  $\lambda^{\alpha}$ with the  RHS of  the formula (\ref{E:lambdaformula}).

We want to finish this section with a question. In general the Poincar\'e lemma fails in a tangential CR complex (see e.g.\cite{Boggess}). Does it fail in $\Omega_{H^{0,1}}(\cP_{\M,F})$?

\section{Reality conditions}\label{S:reality}
The classical 11-D supergravity    is defined over the field of real numbers, whereas we work over the  complex numbers. 
The  missing reality conditions will be formulated in this section.
  
A real analytic  function  $f(z)=\sum_{k=0}^{\infty}c_iz^i,c_i\in \R$ is  characterized  by the  identity $f(z)=\overline{f(\bar{z})}$. More generally, a real analytic   function  $f$ on a complex manifold $X$ equipped with  an  anti-holomorphic involution $\rho$ is characterized by
 \[f(z)=\overline{f(\rho(z}).\] This definition of reality 
 extends to the space of complex smooth  differential forms $\Omega^{k}=\bigoplus_{i+j=k}\Omega^{i,j}=\bigoplus_{i+j=k}\Omega^{i,j}(X)$.
The involution $\rho$ maps  $\omega\in \Omega^{i,j}$
to $\rho^*\omega\in \Omega^{j,i}$. Bear in mind   that  $\overline{\rho^*\omega}\in \Omega^{i,j}$ and $\overline{\rho^{*}\dbar}=\dbar$. A real form
 satisfies
\[\omega=\overline{\rho^*\omega}\] Real forms define a sub-complex in $\bigoplus_{j} \Omega^{i,j}$.
The definition extends to super CR manifolds: a map $\rho$ is a $C^{\infty}$ CR involution   if $\rho^* H^{1,0}\subset \overline {H^{1,0}}$ and $\rho^2=\id$. The role of $(\Omega^{0,p},\dbar)$ is played by  the tangential CR complex, in which $\rho$ defines an anti-linear automorphism  of $\Omega^{0,p}_{H^{0,1}}$.

\section{A cohomological invariant of the CR structure on $\cP$}\label{invariantCRstr}
In this section we develop rudiments of the structure theory of super CR manifolds adapted to the odd twistor space $\cP$. The structure theory of  holomorphic  supermanifolds was studied in \cite{Calib}.   A holomorphic $(n|m)$-dimensional supermanifold  $Y^{}$ has the following basic  invariants: the underlying even $n$-dimensional manifold $Y_{\red}$ and a holomorphic rank $m$  vector bundle $\G$. A more refined invariant is a sequence of characteristic classes $\omega_i$, with the simplest  $\omega_1\in H^{1}(Y_{\red}, \Lambda^2\G^*\otimes T(Y_{\red}))$. Keep in mind that these characteristic classes have no immediate relation  to the topological characteristic  classes of vector bundles. The manifold $Y$ can be thought of as a deformation of the split manifold $Y_{split}=\Pi\G$,  the deformation    is trivial on $Y_{\red}$. The characteristic class $\omega_1$  in $H^{1}(Y_{split}, T(Y_{split}))$ (we interpret sections of $\Lambda^2\G^*\otimes T(Y_{\red})$ as local vector fields on $\Pi\G$) is zero when $Y\cong Y_{split}$. A non zero  $\omega_1$ is an obstruction to splitting of $Y$. The  \v{C}ech approach to cohomology was used in   \cite{Calib} for  the construction of $\omega_1$. Dolbeault cohomology has the same basic functionality, but it is more flexible because it admits a generalization to the CR case. We shall not attempt to  develop a theory of characteristic classes of super CR manifolds in the full generality.  Instead, the  goal of this section  is to  identify the cocycle $\omega_1$ and the group it belongs to in the case of $\cP_{\M,F}$.

In our definition of $\omega_1(\cP)$,  we certainly want to follow  the structure theory of holomorphic supermanifolds outline above.    Obviously, $\cP_{\red}$ is a relative Isotropic Grassmannian  $\cP_{\red}\cong \OGr(2,11)_{M}$  with the projection $\OGr(2,11)_{M}\overset{p}\rightarrow M$. 
The split form $\cP_{split}\cong \Pi p^*S_M$ has  a CR-structure that is nontrivial only on the fibers
of the projection $q_{split}:\cP_{split}\rightarrow M$.
We denote a fiber by $\wOGr(2,11)$.
Then
\begin{equation}\label{E:wOGr}
\wOGr(2,11)\cong\OGr(2,11)\times \Pi \s_{10,1\ \R}
\end{equation}
The subbundle $H^{1,0}\subset T^{\C}(\wOGr(2,11))$ is still defined by formulas (\ref{E:twmap},\ref{E:distgengen}) where   $\{\eta_{\alpha}\}\subset \Pi \s_{11}$ is a basis in the space of the constant spinors. 
We shall define now a collection of cocycles $\{\gamma^i\}$ that are  tangential CR forms over $\wOGr(2,11)_b\cong q_{split}^{-1}(b)$.

The construction of $\{\gamma^i\}$ simplifies  if we present $\dbar_{H^{0,1}}$ as a sum of two anti-commuting differentials $d_{I}$ and $d_{II}$. The differential $\dbar_{H^{0,1}}$ has the bi-degree $(1,1)$ 
with respect to the bigrading $(c,c')$ on  $\Omega_{H^{0,1}}(\cP)$  defined by the rule \[(c,c')=(\deg_{d \bar{a}}f, \deg_{\mu}f),f\in \Omega_{H^{0,1}}(\cP).\]   The  $(1,0)$ component of $\dbar_{H^{0,1}}$ is  $d_{I}=d\bar{a}^{ij}\frac{\sd}{\sd \bar a^{ij}}.$ The $(0,1)$ component is  
\begin{equation}\label{E:dii}
d_{II}=\mu^{A}\left(h(x,\theta,a)^{\alpha}_A\frac{\sd}{\sd \theta^{\alpha}}+g(x,\theta,a)^{i}_{A}\frac{\sd}{\sd x^{i}}\right)
\end{equation}

We need to describe local sections of  $\Omega_{H^{0,1}}(\wOGr(2,11)_b)$ in a more down-to-earth terms.
 If we set $x^i$ to constants $b^i$ ($b=(b^i)$), then the remaining variables (\ref{E:variables}) by definition are (possibly singular) sections of  $\Omega_{H^{0,1}}(\wOGr(2,11)_b)$ (\ref{E:wOGr}). The space of $C^{\infty}$ sections of $\Omega_{H^{0,1}}(\wOGr(2,11)_b)$  is isomorphic to the space of sections of 
\[ \bigoplus_{p,i,j\geq 0}\Omega^{0,p}\Lambda^{i} \s^*_{11}\otimes \Sym^{j}\s^*_{\OGr(2,11)}.\]
Bear in mind that a local section of $\Sym^{j}\s^*_{\OGr(2,11)}$ is a  local holomorphic function on the total space of $\s_{\OGr(2,11)}$ of degree $i$ homogeneity in the fiber-vise direction (see Section \ref{S:structure} for details).
  The differential $d_{I}$ acts on the elements of the algebra generated by $\theta^{\alpha}, a^{ij},\bar{a}^{ij}, \mu^{\rA},d\bar{a}^{ij}$. In general $d_{II}f(\theta^{\alpha}, a^{ij},\bar{a}^{ij}, \mu^{\rA},d\bar{a}^{ij})$ is $x$-dependent,  but if we remove all terms in $d_{II}$ (\ref{E:dii}) of degree one and higher in $\theta$, the remaining differentiation $d'_{II}$ transforms $\Omega_{H^{0,1}}(\wOGr(2,11))$ to itself and squares to zero. By definition $\dbar_{H^{0,1}, \wOGr(2,11)}=d_I+d^{'}_{II}$.
 
We are ready to describe the cocycles $\gamma^i$ in local coordinates.
 In the flat case    $\gamma^i=\mu^{\rA}g(\theta,b)^{i}_{A}=\dbar_{H^{0,1}}x^i=\Gamma^{i}_{\alpha\beta}\theta^{\alpha}\lambda^{\beta},i=1,\dots,11$, with $\lambda^{\beta}$ replaced by (\ref{E:lambdaformula}). Elements $\gamma^i(b)$ are sections of $\Omega_{H^{0,1}}(\wOGr(2,11)_{b})$. They  can be packaged into a single object $\omega_1(\cP)=\gamma^i(x)\frac{\sd}{\sd x^i}\in \Omega_{H^{0,1}}(\cP_{split},T_{CR}(\cP_{split}))$ in which  $x$-dependence is restored. 
 Here $T_{CR}(\cP_{split})$ is $T^{\C}(\cP_{split})/H^{0,1}$. The first cohomology  group of $\Omega_{H^{0,1}}(\cP_{split},T_{CR}(\cP_{split}))$ is an analogue of $H^{1}(Y_{split}, T(Y_{split}))$ in the holomorphic theory. 
  In the non flat case   
 \begin{equation}\label{E:defgamma}
 \mbox{$\gamma^i$  is the leading term of   $\dbar_{H^{0,1}}x^i$   in $\theta$ of degree $\deg_{\theta}=1$.}
 \end{equation} 
  Equation 
  \begin{equation}\label{E:gammacocyleeq}
  (d_I+d'_{II})\gamma^i=0
  \end{equation}
   follows from $\dbar_{H^{0,1}}^2=0$. Some simple properties of $\gamma^i$ are  established in \cite{MovGr}.
 In particular \cite{MovGr} contains a computation of the cohomology of $\Omega_{H^{0,1}}(\wOGr(2,11))$.
  Elements $\gamma^i$ generate $H^1(\Omega_{H^{0,1}}(\wOGr(2,11)))\cong V^{11}$. 
They transform covariantly as vectors under the action of $\Spin(10,1,\R)$. 
   
   Note that if $\dbar_{H^{0,1}}x^i$  were all zero, the manifold $\cP$ would be split. The manifold  $\cP$ would  still  be  split were the elements   $\dbar_{H^{0,1}}x^i$ just cohomologous to zero in an imprecise sense, which takes into account the local coordinate change. This is why $\omega_1$ is  an obstruction to splitting of $\cP$.

 The proof of  Proposition \ref{P:quasiisomorphism} is  simple, provided we take for granted the following result.
 \begin{proposition}\label{E:identification}(cf.\cite{MovGr})
\[\H^0(\OGr(2,11),\Sym^i\s^*_{\OGr(2,11)}))=A_i,\] \[\H^k(\OGr(2,11),\Sym^i\s^*_{\OGr(2,11)}))=0,k\geq 1.\]
\end{proposition}
 The  computation of  $\H^i( \Omega_{H^{0,1}})$  can be done in two stages. The first is the computation of  the $d_I$-cohomology. The resulting algebra $\E$  has less generators then $\Omega_{H^{0,1}}$. The second is the computation of the $d_{II}$-cohomology  of $\E$. Notice that $d_{I}$ does not depend on $x$ and $\theta$.
Elements of the algebra generated by $\mu, a, \bar{a},  d\bar{a}$ are possibly singular $C^{\infty}$ Dolbeault  forms with values in $\bigoplus_{n\geq 0}\Sym^{n}\s^*_{\OGr(2,11)}$. The $d_{I}$-cohomology of the subalgebra  of regular $C^{\infty}$ forms   is  the algebra $A$ (\ref{E:Adef}) (see  Proposition \ref{E:identification}). This explains why the stage two of the procedure  is identical  to the $D$-cohomology computation in $\Grv_{\M,F}$. The reality condition enforced by $\rho$ picks up polynomials in the generators of $A$ with real coefficients.
 
 The  CR structure on $\cP$ is compatible with  the trivial CR structure on the relative  grassmannian $\OGr(2,11)_{M}=\cP_{\red}$. Let $J$ be the kernel of the restriction map $\Omega_{H^{0,1}}(\cP)\rightarrow \Omega_{H^{0,1}}(\cP_{\red})$. The arguments from the previous paragraph become global  in the framework of the spectral sequence associated with the filtration $J^{\times n}$ in $\Omega_{H^{0,1}}(\cP)$. The spectral sequence degenerates in the $E^2$ term  as a consequence of Proposition \ref{E:identification}.
 
\section{The inverse transform}\label{S:inverse}
The odd twistor transform converts a SUGRA datum $(\M,F)$  that satisfies (\ref{E:torsion}, \ref{E:SG}) into a  $(2\times15+11|2\times8+16)$-dimensional CR manifold $\cP=\cP_{\M,F}$. In this section we shall concern ourself with the intrinsic  characterization of  $\cP$ in the class $(2\times15+11|2\times8+16)$-dimensional CR manifolds.
 To summarize our previous discussion, we  list the most important characteristics  of $\cP$:
\begin{enumerate}
\item \label{E:data1} The complexified  tangent bundle $T^{\mathbb{C}}({\cP})$ contains a complex rank $(15|8)$  subbundle $H^{1,0}$ that defines an integrable CR-structure.
\item \label{E:data2} There is a non-empty family $\OGr(2,11)\subset \cP$ of CR-holomorphic Orthogonal Grassmannians. The real normal bundle   is trivial $N_{\OGr(2,11)}\cong \OGr(2,11)\times V^{10,1\ \R}\times \Pi \s^{\mathbb{R}}_{10,1}$. 
The bundle $H^{1,0}|_{\OGr(2,11)}$ is a (trivial) extension 
\begin{equation}\label{E:exact}
0\rightarrow T(\OGr(2,11))\rightarrow H^{1,0}\rightarrow \Pi \bar{\s}_{\OGr(2,11)}\rightarrow 0
\end{equation}
\item \label{E:data2.5} 
The preimage of $\OGr(2,11)\subset \cP_{\red}$ in $\cP_{split}$ as a CR manifold 
is isomorphic to $\wOGr(2,11)$ (\ref{E:wOGr}).

  Let  $x^i, i=1,\dots,11$ be local even  independent functions  on $\M$. By abuse of notations, we denote $p^*x^i$ by $x^i$.
  In the notations of Section \ref{S:structure}, the differential $\dbar_{H^{0,1}}x^i\in \Omega^1_{H^{0,1}}(\cP)$ can locally  be written as 
  \begin{equation}\label{E:gdef111}
  g^i_{A}(x,\theta,a)\mu^{A}.
  \end{equation} Keep in mind that it is automatically independent of  $\bar a$ and  $d\bar a$.   A section $ \gamma^i_{\alpha,A}(b,a)\theta^{\alpha}\mu^{A}$  of $\Omega^1_{H^{0,1}}(\wOGr(2,11)_b)$ is  the leading in $\theta$ term of (\ref{E:gdef111}).
  The classes of $\gamma^i$ define a basis in $\H^1(\Omega_{H^{0,1}}(\wOGr(2,11)))$ (cf. discussion in Section \ref{invariantCRstr}).

\item \label{E:data3} There is a fixed point free involution $\rho:\cP\rightarrow \cP$ that maps $H^{1,0}$ to $H^{0,1}$. At least one  of $\OGr(2,11)$ in the family  is $\rho$-invariant. The involution commutes with a holomorphic $\SO(10,1,\R)$ action  on the $\OGr(2,11)$.
\end{enumerate}\label{D:embeddable}
\begin{definition}A supermanifold $\cP$ is {\it globally embeddable} if  it is a closed CR submanifold of  a complex $(26|24)$-dimensional manifold $\cP_{\mathbb{C}}$. We assume that $\rho$ extends to $\cP_{\mathbb{C}}$ as fixed-point free antiholomorphic involution.
\end{definition}

Global embeddability  of real-analytic CR structures on ordinary manifolds was established by Andreotti and Fredricks \cite{AndreottiandFredricks}. Presumably, their technique  admits a super-extension that can be applied to a real-analytic $\cP$. Meanwhile we just simply assume that $\cP$ is globally embeddable.

We shall describe  how to  construct a super space-time $\M$ from $\cP\subset \cP_{\mathbb{C}}$ that satisfies conditions (\ref{E:data1},\ref{E:data2},\ref{E:data2.5},\ref{E:data3}). 
The idea  goes back to Penrose. We identify $\M$ with the $\rho$-real points in the moduli space $\M^{\C}$ of $\OGr(2,11)\subset  \cP_{\mathbb{C}}$. 

The  existence theorem for the versal family of compact super subvarieties \cite{Vaintrob} relies on vanishing of the cohomology groups associated with the normal bundle of the subvariety.  We begin with a computation of these  groups.
It follows from  (\ref{E:exact}) and Assumption (\ref{E:data2}) that the holomorphic normal bundle $N_{\OGr(2,11)}$ is the quotient of $\OGr(2,11)\times V^{11}\times \Pi \s_{11}$ by  $\Pi {\s}_{\OGr(2,11)}$ (\ref{E:bundledef}).
 The formal tangent space to the moduli of deformations $\OGr(2,11)\subset \cP_{\mathbb{C}}$  is isomorphic to $\H^0(N_{\OGr(2,11)})=\H^0(\OGr(2,11),N_{\OGr(2,11)})$.  The space of obstructions is $\H^1(N_{\OGr(2,11)})$ (cf.\cite{KodairaSubm}). 
 \begin{proposition}\label{P:trivial}
Let $Y$ be a compact projective homogeneous space of a complex semisimple group $G$, with the Lie algebra $\g$.  The nontrivial cohomology of the structure sheaf and the tangent sheaf  are $\H^0(Y,\O)=\mathbb{C}$, $\H^0(Y,T(Y))=\g$ respectively.
\end{proposition}
\begin{proof}
The proof follows from Theorem VII in \cite{bott}.
\end{proof}

 This verifies that  $\H^1(\OGr(2,11),T(\OGr(2,11)))=\{0\}$ and that the space $\OGr(2,11)$ is rigid \cite{KodairaandSpencer}.

 The nonzero even component of the cohomology  is  \[\H^0({\OGr(2,11)},N)^{\mathrm{even}}=V^{11}\](Proposition \ref{P:trivial}).
 We extract the odd part  from the long exact sequence:
\[
\begin{split}
&0\rightarrow \H^0(\Pi {\s}_{\OGr(2,11)})\rightarrow \H^0(\O)\otimes \Pi \s_{11} \rightarrow \H^0(N_{\OGr(2,11)})^{odd}\rightarrow \H^1(\Pi {\s}_{\OGr(2,11)})\rightarrow \\
&\rightarrow \H^1(\O)\otimes \Pi \s_{11} \rightarrow \H^1(N_{\OGr(2,11)})^{odd}\rightarrow \H^2(\Pi {\s}_{\OGr(2,11)})\rightarrow \dots
\end{split}
\]
Vanishing  of $\H^i({\s}_{\OGr(2,11)})$ 
was verified in \cite{MovGr}.
We derive that 
\begin{equation}\label{E:decomp}
\H^0(N)\cong V^{11}+\Pi \s_{11},\quad \H^1(N)=\{0\}
\end{equation} 
  By the super version of the Kodaira theory of deformation of compact immersions \cite{Vaintrob}
    $\OGr(2,11)\subset \cP^{\mathbb{C}}$
can be included in a complex-analytic versal family $\M^{\mathbb{C}}$.   We define $\M$ to be the real locus of the involution $\rho$ in $\M^{\C}$. Note that the condition (\ref{E:data3}) on $\rho$  is rigid \cite{BVinbergALOnishchik}. In principle, $\M$ might have several connected components in $\M^{\C}$. We do not reject a possibility of a disconnected  $\M^{\C}$ either.

Our next task is to define the distribution $F$ (\ref{E:defdist}) that satisfies (\ref{E:SG}).
We construct it using the graph of the universal family $\Q\subset  \cP^{\mathbb{C}}\times \M^{\mathbb{C}}$ (cf. \cite{Calib}).
It fits into the diagram  \[\cP^{\mathbb{C}}\overset{r}{\leftarrow} \Q\overset{p}{\rightarrow} \M^{\mathbb{C}}\]
Leaves of $r$ 
are purely odd affine spaces. A fiber $r^{-1}(x)$ $x\in \OGr(2,11)\subset\cP^{\mathbb{C}}$ is modeled by a subspace of sections of $N_{\OGr(2,11)}$ that vanish at $x$. The $r$-vertical tangent subspaces in $T(\Q)$ under projection $p$ span odd subbundle $F\subset T(\M^{\mathbb{C}})$. By the construction $p(r^{-1}(x))$ is tangential to $F$. On the general grounds  $p^{-1}(\M)$ is a CR submanifold in $\Q$ and $r:p^{-1}(\M)\rightarrow \cP^{\C}$ is a CR map. The complexified real tangent bundle $T^{\C}(p^{-1}(\M))$ is isomorphic to the extension of 
$T(\Q)|_{p^{-1}(\M)}$ by
 the antiholomorphic relative tangent bundle 
 $\overline{T}_{\Q/\M^{\C}}|_{p^{-1}(\M)}$.
 We assume that 
\begin{equation}\label{E:isoPCR}
p^{-1}(\M)\overset{r}\rightarrow \cP\subset \cP^{\C}
\end{equation} is a global isomorphism. Near a real point $b\in \M$  represented by $\OGr(2,11)\subset \cP$ (Condition \ref{E:data2}) a local isomorphism  follows from the inverse function theorem.
  The subbundle $H^{0,1}\subset T^{\C}(p^{-1}(\M))$ is isomorphic to the extension of $T_{\Q/\cP^{\C}}|_{p^{-1}(\M)}$ by $\overline{T}_{\Q/\M^{\C}}|_{p^{-1}(\M)}$. The extension is isomorphic (by the assumption (\ref{E:data2.5})) to (\ref{E:exact}) and  is compatible with the isomorphism (\ref{E:isoPCR}).

It remains to verify the torsion equation (\ref{E:SG}).  Let us choose even independent local holomorphic functions $z^i, i=1,\dots,11$ that vanish at  $z\in \M \subset  \M^{\C}$. We set $x^i=\mathrm{Re} z^i$ 
 and choose a vector  $\xi\in T_w(\Q)$ $w\in p^{-1}(z)$  that  is  tangential to a fiber of  $r$. A vector   $(Dp)\xi\in T_z(\M^{\mathbb{C}})$ is the image   of $\xi$   under the differential $Dp$ of the map $p$. Then 
\[p^*\frac{\sd x^i}{\sd (Dp)\xi} = \frac{\sd p^*x^i}{\sd \xi} =\xi \righthalfcup dp^*x^i=\xi \righthalfcup \dbar p^*x^i =\xi \righthalfcup \dbar_{H^{0,1}}p^*x^i=\xi \righthalfcup g^i=p^*((Dp)\xi\righthalfcup p_*g^i).\]
The term $g^i$ is the same as in  \ref{E:dii}.
In short
\begin{equation}\label{E:formEdefinition}
\frac{\sd x^i}{\sd (Dp)\xi}-(Dp)\xi\righthalfcup p_*g^i=0.
\end{equation}
The vector $\xi$ is a value of a holomorphic vector field that is defined in a small neighborhood of $w$. It commutes with any local, antiholomorphic tangential to fibers of $p$ vector field $\bar\zeta$. This implies that $0=\frac{\sd}{ \sd \xi} \frac{\sd p^*x^i}{\sd \bar\zeta}=\frac{\sd}{\sd \bar\zeta} \frac{\sd p^*x^i}{\sd \xi}$.
From this we conclude that expression   $g^i=g^i_{A}(x,\theta,a)\mu^{A}$ (\ref{E:defgamma}) is a global holomorphic, $x$ and $\theta$-dependent section of $\s^*_{\OGr(2,11)}$ - the bundle dual to $\s_{\OGr(2,11)}$. With the help of Proposition \ref{E:identification} we convert $g^i$ to  $g^i_{A,\alpha}(x,\theta)\lambda^{\alpha}=g^i_{A,\alpha}(x,\theta)d\theta^{\alpha}$ (cf. discussion in Section \ref{S:structure}). Equation \ref{E:formEdefinition} is equivalent to $Dp(\xi)\righthalfcup {E}^i=0$ for the forms $ {E}^i=dx^i-g^i_{A,\alpha}(x,\theta)d\theta^{\alpha}.$ The span $<CX>$ of the image $CX$ of the map $(W,\xi)\rightarrow \xi\in\s_{11}, (W,\xi)\in \s_{\OGr(2,11)}$ (see equation (\ref{E:bundledef}) notations) coincides with $\s_{11}$. 
It implies that  $\eta\righthalfcup E^i=0$ for all $\eta\in F_z$.  The map $p$ is a submersion with purely even fibers, therefore, $\dim F_z=(0,32)$. We conclude that  the  independent forms $E^i, i=1,\dots,11$ define $F_z$.

By definition $\dbar^2_{H^{0,1}}p^*x^i= \dbar_{H^{0,1}}g^i=0$. 
This  implies that the restriction of two-from $dE^i$ on $CX\subset F_z$ is zero. It is known that $CX$ is the space of complex solutions of the equations $v^i=0$ (\ref{E:pure})(see  \cite{MovGr} for details). So $dE^i=c^i_j\Gamma^j_{\alpha\beta}d\theta^{\alpha}d\theta^{\beta}+G^i_{j}E^j$ where $G^i_{j}$ are some one-forms. The proof follows if we prove that  $c^i_j$ is invertible. Invertibility follows from  the conceptually simple homological considerations. We shall be sketchy and leave to the reader to fill in the missing details. First we show that the complex $\Omega_{H^{0,1}}(\wOGr(2,11))$ computes $\Tor^{\Sym [\s_{11}]}(A,\C)$. By Assumption (\ref{E:data2.5} )  classes $\gamma^i$ (leading $\deg_{\theta}=1$ terms of $g^i$ (\ref{E:defgamma})) define a basis in $\Tor^{\Sym [\s_{11}]}_1(A,\C)$. Then we compute the same group using the minimal free $\Sym [\s_{11}]$ resolution of $A$.  We interpret ${\gamma'}^i=c^i_j\Gamma^j_{\alpha\beta}d\theta^{\alpha}d\theta^{\beta}$ as cocycles in  the minimal resolution approach.  To show that the classes  ${\gamma'}^i$ and ${\gamma}^i$ coincide in $\Tor_1^{\Sym [\s_{11}]}(A,\C)$ and the matrix $c^i_j$ is invertible,  we use equivalence of the two approaches.

 The inverse odd twistor transform of $\cP$ provides us with  the SUGRA datum $(\M,F)$.  Note that all the steps of the inverse  transform  are reversible and the direct  transform $\cP_{\M,F}$ is identically equal to $\cP$.

In order to construct a versal family of $\OGr(2,11)$ in $\cP$  that does not have a global complex embedding,   more advanced analytic methods are needed.
\section{An even modification of the CR structure on $\cP$}\label{S:modCR}
There is an interesting  modification of the CR structure on $\OGr(2,11)_{\M}$ evoked by equation (\ref{quastion3}). The modified complex distribution ${H'}^{1,0}\subset T^{\C}(\OGr(2,11)_{\M})$ also fits into the diagram (\ref{E:twmap}). The map $Dp$ has the same kernel. Choose a splitting $T^{\C}_z(\M)=F^{\C}_z+N_z^{\C}$. The space $(Dp){H'}_x^{1,0}$    in addition to (\ref{E:distgengen}) contains a span of 
\begin{equation}
v^i=\bar{a}^{ij}f_{j}
\end{equation}
Vectors $\{f_j\}$  define an orthonormal basis in $N_z$. Note that ${H'}^{1,0}$ does not depends on the splitting.  The proof of integrability of the new CR structure goes through. We denote the resulting CR manifold by $\cR_{\M,F}$. The manifold $\cR$ has dimension $(2\times 17+7|2\times8+16)$. We plan to investigate relation of $\cR_{\M,F}$  to $\cP_{\M,F}$ and to study homological properties of 
$\Omega^{0,i}_{H^{0,1}}(\cR_{\M,F})$ in the following publications.

\bigskip

{\bf\Large  Appendix}
 
\appendix
\section{Useful decompositions of adjoint and spinor  representations}\label{S:adjoint}
The local coordinates $\mu^{\rA}$(\ref{E:coord}),  $a^{ij}$  on the total space of the vector bundle $\s_{\OGr(2,11)}$(\ref{E:bundledef}) depend on the choice of a base point $U\in \OGr(2,11)$. In this section we shall elaborate in this.  
The coordinates and the base point depend on the  direct sum decomposition 
 \begin{equation}\label{E:decomposition1}
 \vv^{11}\cong \vv^7+\vv^4=\vv^7+U+ U^{'}
\end{equation}
of the fundamental vector representation of the complex $\SO(11)$.
$\vv^i$  stands for an $i$-dimensional complex Euclidean space; $V^7\perp V^{4}$. The two-dimensional spaces   $U,U^{'}\subset \vv^4$ such that $U\cap U^{'}=\{0\}$ are  isotropic. The inner product defines a non-degenerate pairing between $U$ and $U^{'}$. The summands in (\ref{E:decomposition1}) are irreducible $\SO(7)\times \GL(2)$ representations.

 By utilizing  (\ref{E:decomposition1}) we immediately arrive at  the   decomposition of  $\Ad\SO(11)\cong\Lambda^2\vv^{11}$.
  it has a form  of a spectral decomposition by   the eigen subspaces of the central element $c\in \gl_2\cong U\otimes U^*\cong U\otimes U'$:
\begin{equation}\label{E:decomp2}
\begin{split}
&\Ad(\mathfrak{so}_{11})_2=\Lambda^2U'\\
&\Ad(\mathfrak{so}_{11})_1=\vv^7\otimes U'\\
&\Ad(\mathfrak{so}_{11})_0=\Lambda^2\vv^7+U\otimes U^{'}\\
&\Ad(\mathfrak{so}_{11})_{-1}=\vv^7\otimes  U\\
&\Ad(\mathfrak{so}_{11})_{-2}=\Lambda^2  U
\end{split}
\end{equation}
The one-dimensional linear space  $\Lambda^2  U$ is the  Pl\"ucker image in $\P(\s_{11})$  of  the already mentioned  base point $U\in \OGr(2,11)$.
The construction of the coordinates $\mu^{\rA}$ relied also on a direct sum decomposition of the spinor representation. We shall see now how this comes about.  The space of (Dirac) spinors $\s_{11}$
 is an irreducible module over the Clifford algebra $Cl(\vv^{11})$.  The spinor representation $\s_{11}$ is symplectic \cite{Deligne}. Let $C$ be the corresponding skew-symmetric $\Spin(11)$-invariant inner product with a matrix  $C_{\alpha\beta}$    in  the basis  $\{\eta_{\alpha}\}\subset \s_{11}$. The decomposition (\ref{E:decomposition1}) explains identifications
\[Cl(\vv^{11})\cong Cl(\vv^7)\otimes Cl(\vv^4)\]
\[\s_{11}\cong \s_{7}\otimes \s_4, \dim_{\C}\s_{11}=32, \dim_{\C}\s_{7}=8,\dim_{\C}\s_{4}=4 \]
The complex spinor representation $\s_4$ is a direct sum $W_l+W_r$ of irreducible two-dimensional representations of  $\Spin(4)\cong\SL(2)\times \SL(2)$.
We arrived at  a $\Spin(7)\times \Spin(4)$-isomorphism 
\begin{equation}\label{E:spinordec}
\s_{11}=\s_7\otimes W_l+\s_7\otimes W_r
\end{equation}

The spinor representation $\s_{11}$ can be further decomposed 
into eigenspaces of the central element $c\in \gl_2$.

\begin{equation}\label{E:grading}
\begin{split}
&s^{1}=\s_7\otimes f^+\\
& s^0=\s_7\otimes W_l\\
&s^{-1}=\s_7\otimes f^-
\end{split}
\end{equation}
where $f^{\pm}$ are  the $c$-eigenvectors in $W_r$. It is this decomposition that is used for constructing coordinates $\mu^{\rA}$ on $\s_{\OGr(2,11)}$ in Section \ref{S:structure}.

We want  to verify statement (\ref{E:Unondegeneracy}). 
For this we need a description of 
 $\s_{11}$  in terms of the Grassmann algebra \cite{Chevalley}.  Let $P$ and $P^{'}$ be  five-dimensional isotropic subspaces in $\vv^{11}$ such that $P\cap P^{'}=0$. Then
 \begin{equation}\label{E:adecomp}
 \vv^{11}=P+P^{'}+\vv^1, P+P^{'}=\vv^{10}, \vv^{10}\perp\vv^1
 \end{equation}
The bilinear form $(\cdot,\cdot)$ defines a pairing between $P$ and $P^{'}$. The group $\GL(5)$ acts  on $P+P^{'}$ preserving $(\cdot,\cdot)$. It acts  trivially on $\vv^1$.  We interpret this action as an embedding 
 \begin{equation}\label{E:embedding}
 \GL(5)\subset \SO(10)\subset \SO(11).
 \end{equation} The spinor representation $\s_{11}$, when it is restricted on the double cover $\widetilde{\GL}(5)$, is isomorphic to 
 $\Lambda(P^{'})\otimes \det^{\frac{1}{2}}$ (see \cite{Cartan}).  We shall  drop $\det^{\frac{1}{2}}$-factor in the formulae to simplify notations. 
 
$\Gamma$-matrices are the matrix coefficients of a nonzero $\Spin(11)$-intertwiner $\Sym^2\s_{11}\rightarrow V^{11}$, which we call a $\Gamma$-map.
The components of   the $\Gamma$-map are $C$-adjoint to  the multiplication map
  $P^{'}\otimes \Lambda^iP^{'}\rightarrow \Lambda^{i+1}P^{'}$,  to the  contraction map $P\otimes \Lambda^{i+1}P^{'}\rightarrow \Lambda^{i}P^{'}$, and to the map $u|_{\Lambda^{i}P^{'}}=(-1)^i\id$. The action of  $\so_{11}\cong \Lambda^{2}V^{11}$ on $\s_{11}$ is defined  in terms of $\Gamma$-maps. Let  
  \begin{equation}\label{E:basisg}
  \{f_i\}\subset V^{11}\subset Cl(V^{11})
  \end{equation} be an orthonormal basis in $V^{11}$,   $\{\eta_{\alpha}\}$ be a basis in $\s_{11}$ which is compatible with the decomposition (\ref{E:grading}). Then   $f_i\times \eta_{\alpha}\overset{\ddef}{=}\Gamma_{\alpha\gamma}^{i} C^{\gamma\beta}\eta_{\beta}$ and \[f_i\wedge f_{j}\times \eta_{\alpha}=\frac{1}{2}(f_i\times (f_{j}\times \eta_{\alpha})-f_j\times (f_{i}\times \eta_{\alpha}))\overset{\ddef}{=}\Gamma_{ \alpha\gamma ij}C^{\gamma\beta}\eta_{\beta}\]  

Let $e,e'$ be a basis in $U$. The element $e\wedge e'\in \sl_2\subset \so_{11}$ is  nilpotent. By the   elementary representation theory of $\sl_2$ the operator in $\s_{11}$ corresponding to $e\wedge e'\in \so_{11}$   defines an isomorphism between $s^{1}$ and $s^{-1}$. The matrix of this operator is $\Gamma_{ \alpha\gamma ij}C^{\gamma\beta}a^{ij}(U)$. This verifies (\ref{E:Unondegeneracy}).

Let $e_1,e_2, e_3,e_4, e_5\in P$ be linearly independent isotropic vectors in $V^{11}$ such that  $e_1,e_2$ span a two dimensional $ W\in \OGr(2,11)$. 
Then    $e_1\wedge e_2=a^{ij}f_i\wedge f_j$ and   $a^{ij}(W)\Gamma_{\beta ij}^{\alpha }$ is a matrix  of the multiplication operator     on $e_1\wedge e_2$ in $\Lambda[e_1,\dots,e_5]\cong \s_{11}$. 
An explicit description  of the fiber of $\s_{\OGr(2,11)}$ (\ref{E:bundledef}) over $W$ is
\begin{equation}\label{E:imageW}
\s_W\cong e_1\wedge e_2\Lambda[e_3,e_4,e_5].
\end{equation}

\section{The Pl\"{u}cker   embedding of $\OGr(2,11)$ }\label{S:Plucker}
In this section we derive equations that characterize the image 
of the classical Pl\"{u}cker embedding of $\OGr(2,11)$ into $\P(\Lambda^{2}V^{11})$.
Let $e_1\wedge e_2\in \Lambda^2 V^{11}$ be such that  $e_1,e_2$ span an isotropic space $W$ in $V^{11}$.
The following  equations  reflects decomposability and isotropy properties of $e_1\wedge e_2$:
\[
\begin{split}
&(e_1\wedge e_2,e_1\wedge e_2)=\frac{1}{2}\left((e_1,e_1)(e_2,e_2)-(e_1,e_2)(e_1,e_2)\right)=0\\
&e_1\wedge e_2\wedge e_1\wedge e_2=0\\
&\frac{1}{4}\left((e_1,e_1)e_2\circ e_2+(e_2,e_2)e_1\circ e_1-2(e_1,e_2) e_1\circ e_2\right)=0
\end{split}
\]
The symmetric product $e_i\circ e_j$ is an element in $\Sym^2 V^{11}$.
We expand $e$, $e'$ in the orthonormal basis: $e=\sum_{i=1}^{11}a_1^if_i,\ e'=\sum_{i=1}^{11}{a_2}^if_i$. The skew-symmetric matrix $a^{ij}_{e_1,e_2}=a_1^i{a_2}^j-{a_2}^ja_1^i$   
\begin{equation}\label{E:Plucker}
e_1\wedge e_2=\sum_{i,j=1}^{11}a^{ij}f_i\wedge f_{j}
\end{equation} 
is a function of the basis.
\begin{equation}\label{E:scalingdegree}
a^{ij}_{Be_1,Be_2}=\det(B) a^{ij}_{e_1,e_2}, B\in \GL(2,\C).
\end{equation} 
 We see that the coefficients  $a^{ij}=a^{ij}_{e_1,e_2}$ have  $\GL(2,\C)$-scaling degree one. In other words $a^{ij}(W)=a^{ij}_{e_1,e_2}$ are projective Pl\"ucker coordinates of the point $W=\mbox{span}(e_1,e_2)\in \OGr(2,11)$.
The matrix $a^{ij}$ satisfies
\begin{equation}\label{E:Pluckercoord}
\sum_{i,j=1}^{11}a^{ij}a^{ij}=0\quad
a^{[ij}a^{kl]}=0\quad
\sum_{k=1}^{11}a^{ki}a^{kj}=0
\end{equation}
 We verify in \cite{MovGr} that \ref{E:Pluckercoord} are defining equation for $\OGr(2,11)$ in $\P(\Lambda^2 V^{11})$.


\section{$\Spin(10,1,\mathbb{R})$ orbits in $\OGr(2,11)$}\label{S:orbits}
The  complex group $\SO(11)$ and  its compact form $\SO(11,\mathbb{R})$ act transitively on  $\OGr(2,11)$. The action of the  Lorentz group $\SO(10,1,\mathbb{R})$ has two orbits, which we identify presently. 
 
 An isotropic two-dimensional space $W\subset V^{11}=V^{10,1\ \mathbb{R}}\otimes \mathbb{C}$ defines a real space $E(W)=(W+\overline{W})\cap V^{10,1\ \mathbb{R}}$. The two numerical invariants of $E(W)$  are  dimension and  signature $(d(W),\tau(W))$. Invariants $(d(W),\tau(W))$ completely characterize an orbit $O_{d(W),\tau(W)}$.
 There are  two orbits  \[\OGr(2,11)=\bigcup O_{d(W),\tau(W)}=O_{4,4}\cup O_{3,2}\] 
 We leave the proof  that $O_{4,2}$ is empty 
  to the reader. Here is the idea. Let $W\in O_{4,2}$ then $E(W)\cong\R^{2}\times \R^{1,1}$, $W=W\cap \R^{2\ \C}\times W\cap \R^{1,1\ \C}$, but $W\cap \R^{1,1\ \C}$ is a complexification of a real isotropic subspace in $\R^{1,1}$. Hence $\dim E(W)=3$.
  
  The stabilizers of  base  points in $O_{4,4}$ and $O_{3,2}$ have Lie algebras  $\fu_2\times \so_{6,1}$  and $\so_2\times \R \times  \so_7\ltimes \R^{9}$.  The orbit $O_{4,4}$ is dense in $\OGr(2,11)$, the orbit $O_{3,2}$ has  the real codimension seven in $\OGr(2,11)$.
 

\section{A Lie algebra description of the flat solution}\label{S:flat}
A CR structure on a homogeneous space has a simple description in terms of the Lie algebra data. We use this idea to characterize  the odd twistor transform $\cP$ of the flat solution of SUGRA.

We start with a reminder of  how to describe  a $G$-invariant CR structure on a homogeneous space $G/L$ in the Lie algebra terms(see \cite{AzadandHuckleberryandRichthofer} for more details).  
We assume that $G$ is connected. Left-invariant complex vector fields that  belong to the Lie  subalgebra $\mathfrak{p}\subset\mathfrak{g}^{\mathbb{C}}=Lie(G)^{\mathbb{C}}$  define an involutive  subbundle $H^{1,0}$ in $T^{\mathbb{C}}({G})$.  Let us assume that 
\begin{equation}\label{E:norm}
g\mathfrak{p}g^{-1}\subset\mathfrak{p},g\in St
\end{equation}
where $St$ is the stabilizer group of a base point.
 The subbundle $H^{1,0}$ is invariant with respect to the right $St$-translations and   we can push  $H^{1,0}$ to $G/St$. By construction $H^{1,0}_{G/St}$ is involutive. Obviously, all $G$-invariant involutive distribution in  $T^{\mathbb{C}}({G/St})$ can be obtained this way.  The CR structure is nondegenerate if $\p\cap \overline{\p}\subset \mathfrak{st}^{\C}$. This construction has a straightforward generalization to supergroups.

In the flat case the distribution $F$ on  $\R^{11|32}$ is spanned by (\ref{E:vect}), and $\cP$ is equal to (\ref{E:coset}). 
 
 The super-group of symmetries of the SUGRA datum $(\M,F)$ acts by CR transformations of $\cP_{\M,F}$. 

The Lie algebra of $SP$ is
\begin{equation}\label{E:sp}
\so_{10,1}\ltimes \susy
\end{equation}
where $\susy$ is the algebra of supersymmetries. 
As a linear space it is a direct sum of  $\vv^{10,1 \R}$ (the even part) and  $\s_{10,1 \R}$ (the odd part). The only nontrivial bracket is defined by the formula $[\theta,\theta']=\Gamma(\theta,\theta'),\theta,\theta'\in \s_{10,1}$.
The space $\R^{11|32}$ is the group super-scheme corresponding to Lie algebra $\susy$. Vector fields  (\ref{E:vect}) is a basis in the space of odd left-invariant vector fields on $\R^{11|32}$ (see e.g. \cite{DF}). This explains why $F$ is invariant under left $\susy$ translations and infinitesimal rotations by $\so(10,1,\R)$. These symmetries generate the Lie algebra of $SP$.

We conclude that  the super-Poincar\'{e} group $SP$ acts on $\cP$.Vector fields (\ref{E:distgengen}) for any given $a^{ij}(W)$ span an abelian  Lie subalgebra in $\susy$ (cf. formulae \ref{E:commutator},\ref{correctnessofhom}).

The number of orbits of  $SP$  in $\cP$ coincides the number of orbits   $\SO(10,1,\R)$  in $\OGr(2,11)$.  The Lie algebra of the stabilizer  $St$ of the dense orbit $O$ is isomorphic to $\fu_2\times \so_{6,1}$.

The complex  Lie  algebra $\p$, which describes the CR structure on $O$, is isomorphic to the semidirect product $\p_2\ltimes \Pi \t$ where 
\begin{equation}\label{E:pstab}
\p_2=\Ad(\mathfrak{so}_{11})_0+\Ad(\mathfrak{so}_{11})_{-1}+\Ad(\mathfrak{so}_{11})_{-2} \mbox{ see (\ref{E:decomp2}) for the notations}
\end{equation}
and 
$\t=s^{-1}$ is as in  (\ref{E:grading}). It coincides with the span of (\ref{E:distgengen}) when $a^{ij}(W)=a^{ij}(U)$. The linear space $U$ is the same as in Appendix \ref{S:adjoint}.

The space $\cP_{\mathbb{R}^{11|32}\ \mathbb{C}}$ is a homogeneous space of the complexified super-Poincar\'{e} group $SP$. The isotropy subalgebra $\p\subset \so_{11}\ltimes \susy=Lie(SP)$ of a base point $x\in \cP_{\mathbb{R}^{11|32}\ \mathbb{C}}$ is $\p_2\ltimes \Pi \t$ . 


\begin{thebibliography}{10}

\bibitem{AndreottiandFredricks}
A.~Andreotti and G.A. Fredricks.
\newblock Embeddability of real analytic \uppercase{C}auchy-\uppercase{R}iemann
  manifolds.
\newblock {\em Ann. Scuola Norm. Sup. Pisa Cl. Sci. (4)}, 6(2):285--304, 1979.

\bibitem{AzadandHuckleberryandRichthofer}
H.~Azad, A.~Huckleberry, and W.~Richthofer.
\newblock Homogeneous \uppercase{CR} - manifolds.
\newblock {\em J. Reine Angew. Math.}, 358:125--154., 1985.

\bibitem{BergshoeandSezginandTownsend}
E.~Bergshoeff, E.~Sezgin, and P.~Townsend.
\newblock Supermembranes and eleven-dimensional supergravity.
\newblock {\em Phys. Lett.}, B189(75), 1987.

\bibitem{MemBerkovits}
N.~Berkovits.
\newblock Towards covariant quantization of the supermembrane.
\newblock {\em J. High Energy Phys.}, 0209(051), 2002.

\bibitem{BerkovitsandHoweint}
N.~Berkovits and P.S. Howe.
\newblock The cohomology of superspace, pure spinors and invariant integrals.
\newblock {\em JHEP}, 0806(046), 2008.

\bibitem{Boggess}
A.~Boggess.
\newblock {\em CR Manifolds and the Tangential Cauchy-Riemann Complex}.
\newblock Studies in Advanced Mathematics. CRC Press, 1 edition, 1991.

\bibitem{BonoraandPastiandTonin}
L.~Bonora, P.~Pasti, and M.~Tonin.
\newblock Superspace formulation of 10-d sugra+sym theory a la green-schwarz.
\newblock {\em Phys. Lett. B}, 188:335, 1987.

\bibitem{bott}
R~Bott.
\newblock Homogeneous vector bundles.
\newblock {\em Annals of Mathematics}, 66, No. 2, Sep., 1957(2):203--248, Sep
  1957.

\bibitem{BrinkHowe}
L.~Brink and P.~Howe.
\newblock Eleven-dimensional supergravity on the mass-shell in superspace.
\newblock {\em Phys. Lett. B}, 91:384, 1980.

\bibitem{Cartan}
E.~Cartan.
\newblock {\em The theory of spinors}.
\newblock Dover, New York, 1981.

\bibitem{Cederwall}
M.~Cederwall.
\newblock D=11 supergravity with manifest supersymmetry.
\newblock {\em Mod.Phys.Lett.A}, 25:3201--3212, 2010.

\bibitem{Chevalley}
C.~Chevalley.
\newblock {\em The algebraic theory of spinors and clifford algebras}, volume~2
  of {\em Collected Works}.
\newblock Springer, 1997.

\bibitem{CremmerFerrara}
E.~Cremmer and S.~Ferrara.
\newblock Formulation of eleven-dimensional supergravity in superspace.
\newblock {\em Phys. Lett. B}, 91:61, 1980.

\bibitem{Deligne}
P.~Deligne.
\newblock Notes on spinors.
\newblock In P.~Etingof P.~Deligne, D.~Kazhdan, editor, {\em Quantum fields and
  strings. A course for mathematicians}, volume~1. AMS, 1999.

\bibitem{DF}
P.~Deligne and D.~Freed.
\newblock Supersolutions.
\newblock In P.~Etingof P.~Deligne, D.~Kazhdan, editor, {\em Quantum fields and
  strings. A course for mathematicians}, volume~1. AMS, 1999.

\bibitem{DuffandoweandInamiandStelle}
M.~Duff, T.~Inami P.~Howe, and K.~Stelle.
\newblock Superstrings in d=10 from supermembranes in d=11.
\newblock {\em Phys. Lett.}, 191B(70), 1987.

\bibitem{ECremmerBJuliaandJScherk}
B.~Julia E.~Cremmer and J.~Scherk.
\newblock Supergravity theory in eleven-dimensions.
\newblock {\em Phys. Lett. B}, 76:409--412, 1978.

\bibitem{HoweandHartwell}
P.S. Howe and G.~G. Hartwell.
\newblock A superspace survey.
\newblock {\em Classical Quantum Gravity}, 12(8):1823--1880., 1995.

\bibitem{HughesandLiuandPolchinski}
J.~Hughes, J.~Liu, and J.~Polchinski.
\newblock Supermembranes.
\newblock {\em Phys. Lett.}, 180B(370), 1986.

\bibitem{KodairaSubm}
K.~Kodaira.
\newblock A theorem of completeness of characteristic systems for analytic
  families of compact submanifolds of complex manifolds.
\newblock {\em Annals of Mathematics}, 75(1), 1962.

\bibitem{KodairaandSpencer}
K.~Kodaira and D.C. Spencer.
\newblock On deformations of complex analytic structures. i, ii.
\newblock {\em Ann. of Math. (2)}, 67:328--466, 1958.

\bibitem{MasonandSkinner}
L.J.Mason and D.Skinner.
\newblock An ambitwistor \uppercase{Y}ang-\uppercase{M}ills
  \uppercase{L}agrangian.
\newblock {\em Phys. Lett. B}, 636(60--67), 2006.

\bibitem{CNT}
B.E.W.~Nilsson M.~Cederwall and D.~Tsimpis.
\newblock Spinorial cohomology and maximally supersymmetric theories.
\newblock {\em J. High Energy Phys.}, 0202(009), 2002.

\bibitem{Calib}
Y.I. Manin.
\newblock {\em Gauge Field theory and Complex Geometry}.
\newblock Springer-Verlag, New York, 1988.

\bibitem{MArkl}
M.~Markl.
\newblock Transferring \uppercase{A}$\ity$ (strongly homotopy associative)
  structures.
\newblock {\em Rend. Circ. Mat. Palermo (2) Suppl.}, 79:139--151, 2006.

\bibitem{MovshevQ}
M.V. Movshev.
\newblock \uppercase{Y}ang-\uppercase{M}ills theory and a superquadric.
\newblock In {\em Algebra, arithmetic, and geometry: in honor of Yu. I. Manin.
  Vol. II}, volume 270 of {\em Progr. Math}, pages 355--382. Birkh\"{a}user
  Boston, Inc., Boston, MA, 2009.

\bibitem{MovGr}
M.V. Movshev.
\newblock Geometry of a desingularization of eleven-dimensional gravitational
  spinors.
\newblock http://arxiv.org/abs/1105.0127, 2011.

\bibitem{CGNN}
M.~Cederwall U. Gran~M. Nielsen and B.E.W. Nilsson.
\newblock Manifestly supersymmetric \uppercase{M}-theory.
\newblock {\em J. High Energy Phys.}, 0010(041), 2000.

\bibitem{RoslyandSchwarz}
A.A. Rosly and A.S. Schwarz.
\newblock Supersymmetry in a space with auxiliary dimensions.
\newblock {\em Comm. Math. Phys.}, 105(4):645--668, 1986.

\bibitem{Vaintrob}
A.~Yu. Vaintrob.
\newblock Deformations of complex structures on supermanifolds. (russian).
\newblock {\em Funktsional. Anal. i Prilozhen.}, 18(2):59--60, 1984.

\bibitem{BVinbergALOnishchik}
E.~B Vinberg and A.~L. Onishchik.
\newblock {\em Seminar on Lie groups and algebraic groups}.
\newblock Springer series in Soviet mathematics. Springer Verlag, Berlin, New
  York, 1990.

\bibitem{WardandWells}
R.~S. Ward and R.~O.~Wells Jr.
\newblock {\em Twistor Geometry and Field Theory}.
\newblock Cambridge Monographs on Mathematical Physics. Cambridge University
  Press, 1991.

\end{thebibliography}
\end{document}